\documentclass[a4paper, 12pt]{article}
\usepackage{a4wide} 
\usepackage{graphicx}
\usepackage{pgfplots}
\usepackage{here, amsmath, latexsym, amssymb, bm, ascmac, mathtools, multicol, tcolorbox, subfig, amsthm, natbib}

\mathtoolsset{showonlyrefs}
\theoremstyle{definition}
\newtheorem{axiom}{Axiom}
\newtheorem{definition}{Definition}
\newtheorem{theorem}{Theorem}
\newtheorem{proposition}{Proposition}

\newtheorem{corollary}{Corollary}

\newtheorem{lemma}{Lemma}

\newtheorem{remark}{Remark}

\usepackage[colorlinks=true, linkcolor=blue, citecolor=blue]{hyperref}

\title{An axiomatic model of robust Bayesian persuasion\thanks{The authors are grateful to Youichiro Higashi for invaluable discussion and suggestions as the discussant at 2026 Japanese Economic Association Autumn Meeting.
We would also like to thank Daisuke Hirata, Daiki Kishishita, Takeshi Murooka, Nozomu Muto, Satoshi Nakada, Shinpei Noguchi, Shunya Nuta, Yuta Takahashi, Norio Takeoka, Shuichi Tsugawa, Tomoki Tsujita, Eisuke Uchida, Yuya Wakabayashi, Takuro Yamashita, and the participants at Hitotsubashi Economic Theory Workshop, Asian Game Theory Conference, Ryukoku University Economic Society Seminar, and the 2026 Japanese Economic Association Autumn Meeting for their helpful comments.}}
\author{Wataru Kitano\thanks{Graduate School of Management, Tokyo University of Science, 1-11-2 Fujimi, Chiyoda-ku, Tokyo 102-0071, Japan, wataru.kitano9998@gmail.com} and Shohei Yanagita\thanks{The Department of Economics, Takasaki City University of Economics, 1300 Kaminamie, Takasaki, Gumma 370-0801, Japan, shoheiyanagita@gmail.com}}

\date{\today}

\makeatother

\usepackage{pgfplots}
\pgfplotsset{compat=1.18}

\begin{document}

\maketitle

\begin{abstract}
We develop an axiomatic model of robust Bayesian persuasion where the sender cannot fully control the information available to the receiver. After selecting an information structure, the sender expects that more informative structures might be implemented. We model this by allowing the sender to assess each information structure under worst-case information leakage, represented by a set of more informative structures. The model encompasses a wide range of examples of information leakage, which we also explore.
\end{abstract}

{\bf Keywords:} Bayesian persuasion; Robustness; Blackwell informativeness

{\bf JEL classification:} D81, D83, D91

\newpage

\section{Introduction}

\subsection{Motivation}
There are many situations where a principal shapes the information available to an agent.
In the literature, such a situation is often captured by the Bayesian persuasion (henceforth BP) model, established by \cite{kamenica2011bayesian}.
In the BP model, the principal chooses an information structure, which is a system that sends the agent a stochastically generated signal that follows a state-dependent distribution.

One of the key features of the basic BP model is the principal's optimistic conjecture that the information is fully controllable.
That is, each generated signal is sent to the agent without distortion, and the agent never acquires additional information from outside information sources.

However, in some contexts, it is more natural to suppose that the information is not fully controllable and the principal is aware of it.
In this paper, we address the problem of what degree of pessimism should be incorporated into the BP model to explain the principal's preference over information structures.
In particular, we axiomatize the \textit{robust} BP model, where a pessimistic principal faces ambiguity about the content of the additional information available to the agent.
While the worst-case optimal information structure for the principal in the BP model with such ambiguity is characterized by \cite{dworczak2022preparing}, we instead provide a decision-theoretic foundation for the representation including such a model.

We follow the approach of \cite{jakobsen2021axiomatic}, who axiomatizes the basic BP model representation of the principal's preference over information structures (for each fixed menu available to the agent).
In our representation, the admissible set of additional information considered by the principal is identified by taking a richer primitive than that in \cite{jakobsen2021axiomatic}:
 the preference over pairs of the information structure and the menu.
Therefore, our analysis is applicable to the following motivating examples, where the principal shapes both the information and the menu available to the agent.
\begin{enumerate}
\item The first example builds on 
\textit{the behavioral interpretation} proposed by \cite{jakobsen2021axiomatic}, where a principal treats the future self as an agent. 
Suppose an e-commerce platform user consumes nothing currently but anticipates that the future self faces a decision problem of what products or services to purchase.
The current self and the future self are interpreted as the principal and the agent, respectively.
That is, the current user customizes the platform's recommendation system by inputting budgets and tastes such as dietary restrictions to narrow down the consumption menu faced by the future self.
Alternatively, the current user may commit to a particular consumption menu by choosing a specialized platform that serves a specific market segment, which is the modern reinterpretation of the restaurant-choice examples in \cite{kreps1979representation} and \cite{gul2001temptation}.
Moreover, the user can shape the information available to the future self by customizing news feeds on social media.\footnote{This is adapted from the example in \cite{jakobsen2021axiomatic}, which takes an individual's choice over information as primitive.}
Therefore, the consumer shapes both the menu and information available to the future self.
\item The second example builds on \textit{the persuasion interpretation}, also proposed by \cite{jakobsen2021axiomatic}, where a principal and an agent are distinct players. 
Consider seller-consumer interactions mediated by a marketplace platform. 
In this situation, the platform and the seller can be interpreted as the principal and the agent, respectively.
That is, there are two channels through which the platform influences the seller's decision: control over the product lineup and the consumer data disclosed to the seller.
For instance, Apple mediates the market for apps through (i) app lineup regulation that incorporates both general requirements such as safety and brand-specific requirements such as design standards and (ii) control over consumer data disclosed to developers.\footnote{See \url{https://developer.apple.com/app-store/review/guidelines/} (accessed June 24, 2026) for details on app regulation and \url{https://www.apple.com/legal/privacy/data/en/app-analytics/} (accessed June 24, 2026) for details on data disclosure.}
These intermediation policies influence both the admissible app lineup and the information disclosed to developers.
\end{enumerate}

Importantly, in both examples, it is also natural to suppose the principal's control to be imperfect, particularly with respect to agents' information acquisition.
In the first example, the consumer may be unable to control the additional information that the future self receives from recommendation algorithms on social media.
In the second example, the platform may be unable to control a developer's independent research into consumer demand, and such research generates deeper insights when combined with the data provided by the platform.
Moreover, in such cases, the principal may be unable to predict the content of the additional information.

\subsection{Model overview}

We consider a sender who chooses both a menu of acts and an information structure.
The receiver observes a signal and chooses an act from the menu. Our key departure from the standard Bayesian persuasion framework is that the sender
cannot fully control the information available to the receiver. After choosing an
information structure $\sigma$, the sender anticipates that a more informative
information structure may be implemented instead. This captures situations in which the receiver has access to additional information sources, such as
platform-generated recommendations, fact-checking, or future information acquisition.

To model this, we introduce a function $\mathcal{L}$ that assigns to each information structure $\sigma$ a set $\mathcal{L}(\sigma)$ of more informative information structures. When the sender chooses $\sigma$, she evaluates outcomes by taking the worst case over $\mathcal{L}(\sigma)$. Formally, the sender's preference is represented by a robust Bayesian persuasion representation of the form
\begin{equation}
    V(A,\sigma)=\inf_{\pi\in \mathcal{L}(\sigma)} \sum_{s \in \pi} \left[ \min_{f^s \in c^s (A)} \sum_{\omega \in \Omega} v (f_\omega) s_\omega \nu_\omega \right],
\end{equation}
where $v$ is a von Neumann–Morgenstern utility function and $\nu$ is a prior belief. Here, $c^s(A)$ denotes the set of the receiver's optimal choices from the menu $A$ given signal $s$.

This representation reflects two sources of robustness. First, the sender evaluates each information structure under worst-case information leakage captured by $\mathcal{L}(\sigma)$. Second, when the receiver has multiple optimal actions, tie-breaking is assumed to be resolved in a way that is unfavorable to the sender. The receiver is assumed to be a standard Bayesian decision maker: she updates her beliefs according to Bayes' rule and chooses an act that maximizes expected utility.

Our main result provides an axiomatic characterization of this robust Bayesian persuasion representation. The axioms capture (i) standard rationality and expected utility under no information, (ii) pessimism in tie-breaking, and (iii) a novel condition---menu-independent leakage---which requires that the sender's concerns about information leakage do not depend on the chosen menu.

\subsection{Related literature}
This paper contributes to the growing literature on axiomatic characterizations of the BP model \citep{jakobsen2021axiomatic,higashi2025axiomatic,jakobsen2025temptation,mensch2025revealed}.
The existing literature axiomatizes the basic BP model as a representation of various primitive: the principal's preference over information structures \citep{jakobsen2021axiomatic}, the principal's preference over menus \citep{higashi2025axiomatic}, the agent's preference over menus \citep{jakobsen2025temptation}, and the agent's random choice data \citep{jakobsen2025temptation,mensch2025revealed}.
Among them, ours is the first attempt to axiomatize a variant of the BP model where the principal may fail to fully control the information.
Moreover, we show that replacing some of our axioms also characterizes the basic BP model, which complements the literature.

Several papers have characterized information structures through preferences over menus. \cite{dillenberger2014theory} developed a framework that identifies the information structure subjectively perceived by a decision maker. \cite{de2017rationally} proposed a model in which the decision maker chooses an information structure at a cost and provided its axiomatic characterization. \cite{higashi2025axiomatic} modeled Bayesian persuasion and several of its generalizations and characterized each of them axiomatically. In all of these models, information structures emerge as endogenous parameters. A key difference between their models and ours is that we treat information structures themselves as primitive objects. This distinction allows us to analyze not only decisions induced by information, but also attitudes toward information itself, such as information avoidance.

Although axiomatizations of variants of the BP model are underdeveloped, characterizations of principal-optimal information structures for such models are well studied.
In particular, the existing extensions incorporate the possibility that the agent may access multiple information sources \citep{laclau2016public,kolotilin2017persuasion,gentzkow2017competition,kolotilin2018optimal,guo2019interval,bizzotto2020testing,hu2021robust,li2021sequential,dworczak2022preparing,kosterina2022persuasion,matyskova2023bayesian}, the information may be reduced before reaching the agent \citep{lipnowski2020attention,bloedel2021persuading,wei2021persuasion,lipnowski2022optimal,arieli2022bayesian}, or both of them
\citep{dall2024persuading,mylovanov2024constructive}.
Among them, \cite{dworczak2022preparing} are most relevant to our analysis:
they characterize a worst-case optimal solution for the principal in the BP problem with ambiguity about the outside information sources available to the agent.
As an alternative representation, we also characterize the BP model where a rationally inattentive agent may ignore some of the information, as in \cite{lipnowski2020attention}.

As we incorporate ambiguity about which information structure is realized, the literature on BP has also addressed the problems with incredible information structures.
\cite{beauchene2019ambiguous} extend the BP problem to allow the principal to choose a set of information structures, from which it is ambiguous which one will be realized.
Another research stream focuses on the lack of the principal's commitment power \citep{min2021bayesian,frechette2022rules,lipnowski2022persuasion}: with some probability, the principal can manipulate a message generated by an information structure.
Such an ex-post manipulation is distinguished from false messages that may be generated by the committed information structure, although the possibility of false messages in the latter case is also studied \citep{ederer2022bayesian}.
Unlike their settings where only agents are uncertain about the actual information structure, we focus on the situation where the principal cannot completely anticipate the content of the information sent to the agent.


Uncertainty about the information structure is also addressed in the literature outside BP:
\cite{wang2024informativeness} extends Blackwell’s informativeness order to ambiguous experiments, modeled as mappings from an auxiliary state space to the set of unambiguous experiments. \cite{payro5638094updating} provide an axiomatic analysis of belief updating when the decision maker cannot identify which information structure generated the observed signal.
\cite{noor2026belief} axiomatize the belief updating rule such that a decision maker has a subjective belief about what information structure actually generates the signal.
Our paper shares with these studies a setting in which the information structure is not uniquely specified, but focuses specifically on the sender's preference in such a setting.

The motivating examples shown above are inspired by the literature on various information environments.
In the behavioral interpretation example, the future self may be exposed to additional information due to the temptation to acquire biased news \citep[see][for surveys]{gentzkow2015media,puglisi2015empirical} or a platform's recommendations that may not optimize user utility \citep[e.g.,][]{acemoglu2024model}.
Our analysis incorporates the principal's informational choice with such a concern about the additional information that can be provided by adversarial senders.

The persuasion interpretation example is related to the literature on the economics of privacy.
Recent studies focus on the effect of privacy regulations such as (i) restricting the amount of data the platform can collect \citep[see, e.g.,][]{eilat2021bayesian,chen2026role} and (ii) allowing the consumer to choose what kind of data is disclosed to the platform \citep[see, e.g.,][]{ichihashi2020online}.
The literature on consumers' privacy-choice has discussed a puzzle known as the \textit{privacy paradox}: individuals with privacy concerns tend to disclose more data than expected \citep[see][for surveys]{acquisti2016economics,goldfarb2023economics,dube2025frontiers}.
Our analysis provides a novel interpretation for this behavior, which is discussed in Section \ref{Ch_privacy}.


\section{Model}

In this section, we provide our setting and define the model called \textit{robust Bayesian persuasion representation}.

\subsection{Setting}

Our domain enriches the domain considered in \cite{jakobsen2021axiomatic}. There is a finite state space $\Omega$. Denote by $\Delta(\Omega)$ the set of probability distributions over $\Omega$. For any $\mu\in\Delta(\Omega)$ and $\omega\in\Omega$, we denote by $\mu_\omega\in[0,1]$ the probability that $\mu$ assigns to $\omega$. Let $X$ be a finite set of outcomes, and $\Delta(X)$ be the set of lotteries over $X$. An \textit{act} is a mapping $f:\Omega\rightarrow\Delta(X)$. That is, each act maps states to lotteries. Denote by $\mathcal{F}$ the set of acts. For any $f\in\mathcal{F}$ and $\omega\in\Omega$, we denote by $f_\omega\in\Delta(X)$ the lottery that $f$ assigns to $\omega$. A \textit{menu} is a nonempty and finite set of acts. The set of menus is denoted by $\mathcal{A}$. The mixture operation between two acts $f$ and $g$ is defined as usual: for any $\alpha\in (0,1)$, $\alpha f+(1-\alpha)g:=(\alpha f_\omega +(1-\alpha)g_\omega)_{\omega\in\Omega}$. Also, for any $\alpha\in(0,1)$ and $A,B\in\mathcal{A}$, we define the mixed menu $\alpha A+(1-\alpha)B$ as $\alpha A+(1-\alpha)B:=\{\alpha f+(1-\alpha)g\mid f\in A,g\in B\}$. With abusing notation, for any $f\in\mathcal{F}$, a singleton menu $\{f\}\in\mathcal{A}$ is identified as the act $f$. Similarly, for any $p\in\Delta(X)$, a constant act $g_\omega=p$ for all $\omega\in\Omega$ is identified as the lottery $p$.

An \textit{information structure} is a matrix $\sigma$ with entries in $[0,1]$, exactly $|\Omega|$ rows, finitely many columns, no column consisting only of zeros, and the sum of entries in each row is $1$. Denote by $\mathcal{E}$ the set of information structures. A signal $s$ is a $|\Omega|$ dimensional vector whose entries are in $[0,1]$ and at least one of them is non-zero. That is, $s=(s_\omega)_{\omega\in\Omega}\in[0,1]^\Omega$ and there exists $\omega\in\Omega$ such that $s_\omega \neq 0$. Denote by $S$ the set of signals. Every signal can be interpreted as a column of some information structure. The notation $s\in\sigma$ means that a signal $s$ is one of the columns of an information structure $\sigma$. For any $\sigma\in\mathcal{E}$ and $\alpha\in(0,1)$, let $\alpha\sigma$ be the matrix constructed by multiplying each entry of $\sigma$ by $\alpha$. For any $\sigma,\sigma'\in\mathcal{E}$ and $\alpha\in(0,1)$, let $\alpha\sigma\cup(1-\alpha)\sigma'$ be the matrix constructed by appending an information structure $(1-\alpha)\sigma'$ to the right of $\alpha\sigma$.

For any $\sigma,\sigma'\in\mathcal{E}$, we say that $\sigma$ is \textit{Blackwell more informative} than $\sigma'$ if there exists a Markov matrix $B$ such that $\sigma'=\sigma B$. That is, $\sigma'$ is a garbling of $\sigma$.\footnote{For the details about the Blackwell informativeness, see, for example, \cite{blackwell1953equivalent} and \cite{green2022two}.} Denote by $\sigma\trianglerighteq\sigma'$ the relation that $\sigma$ is Blackwell more informative than $\sigma'$. In particular, let $o$ denote a particular null information structure. That is, it is an $|\Omega| \times 1$ matrix. This information structure reveals no information about the realized state. When the sender chooses not to provide any information to the receiver, this can be interpreted as selecting this information structure. Clearly, for any information structure $\sigma$, $\sigma\trianglerighteq o$ holds.

Our model has two primitives.
The first is a preference relation $\succsim$ over the set of pairs of a menu and an information structure. The asymmetric and symmetric parts of $\succsim$ are denoted by $\succ$ and $\sim$, respectively. 
The second is a choice correspondence $c^s$, for each $s\in S$, which assigns to each menu $A\in\mathcal{A}$ a nonempty subset $c^s(A)\subseteq A\setminus\{\emptyset\}$.

\subsection{Robust Bayesian persuasion representations}

Under the above setting, we define the robust Bayesian persuasion representation. Before that, we introduce a useful notation that will be used repeatedly in what follows.
Given a belief $\nu\in\Delta(\Omega)$, an information structure $\sigma$, and a signal $s\in \sigma$, denote by $\nu^s$ be the Bayesian posterior of $\nu$ conditional on $s\in\sigma$. That is, for each $\omega\in\Omega$, 
\begin{equation}
    \nu^s_\omega=\frac{\nu_\omega s_\omega}{\sum_{\omega'\in\Omega} \nu_{\omega'}s_{\omega'}}.
\end{equation}

The robust Bayesian persuasion representation is defined as follows:
\begin{definition}
\label{def_RBP}
    A \textbf{robust Bayesian persuasion representation for $(\succsim,c)$} is a function 
    \begin{equation}
        V(A,\sigma)=\inf_{\pi\in \mathcal{L}(\sigma)} \sum_{s\in\pi,\omega\in\Omega} \underline{V} (\nu ^s) s_\omega \nu_\omega,
    \end{equation}
    where
    \begin{equation}
        \underline{V} (\nu ^s)=\min_{f^s \in c^s (A)} \sum_{\omega \in \Omega} v (f_\omega) \nu^s_\omega
    \end{equation}
    is the sender's interim expected utility after the realization of the signal $s$, $v:\Delta(X)\rightarrow\mathbb{R}$ is a non-constant and mixture-linear function, $\nu\in\Delta(\Omega)$ has full support, and $\mathcal{L}:\mathcal{E}\rightarrow 2^{\mathcal{E}}$ satisfies the following two conditions: (i) $\mathcal{L}(\sigma)\subseteq \{\sigma'\in\mathcal{E}~|~\sigma'\trianglerighteq\sigma\}$ for all $\sigma\in\mathcal{E}$ and (ii) $\mathcal{L}(o)=\{o\}$.
\end{definition}

The functional form of the robust Bayesian persuasion model is motivated by \cite{dworczak2022preparing}. The robust Bayesian persuasion representation is characterized by three parameters. First, $v$ is a von Neumann–Morgenstern utility function representing the sender's preferences over lotteries. Second, $\nu$ is the sender's prior belief over the state space. The novel component is a function $\mathcal{L}$ that assigns to each information structure a set of information structures. For each $\sigma$, $\mathcal{L}(\sigma)$ specifies a set of information structures that are more informative than $\sigma$.

As the representation makes clear, when the sender chooses an information structure $\sigma$, she takes into account not only $\sigma$ itself but also all information structures in $\mathcal{L}(\sigma)$ as potential implementations. The sender evaluates the pair $(A, \sigma)$ based on the worst-case realization within this set. In other words, she anticipates that, for some reason, an information structure more informative than $\sigma$ may be implemented, and that the receiver will make decisions based on the signals generated by such a structure. Moreover, when the receiver has multiple optimal choices, the sender assumes that tie-breaking is resolved in a way unfavorable to her. Thus, the sender evaluates the pair $(A, \sigma)$ robustly with respect to both information leakage and adverse tie-breaking.

Finally, we impose the condition $\mathcal{L}(o) = \{o\}$. 
This assumption is natural when additional information takes the form of ex-post verification: After receiving a signal, the receiver may verify its accuracy, seek a second opinion, or consult another expert. Such verification may replace the sender's information structure with a more informative one. By contrast, when the sender chooses $o$, the receiver receives no signal to verify. Hence, no additional information can arise through this channel.

Notice that the robust Bayesian persuasion representation can be rewritten as
\begin{equation}
    V(A,\sigma)=\inf_{\pi\in \mathcal{L}(\sigma)} \sum_{s \in \pi} \left[ \min_{f^s \in c^s (A)} \sum_{\omega \in \Omega} v (f_\omega) s_\omega \nu_\omega \right].
\end{equation}
Definition \ref{def_RBP} is closer in form to the model introduced by \cite{dworczak2022preparing}. By contrast, the reformulated version above is simpler and preserves the intuition of the original functional form. We therefore treat these two representations interchangeably throughout the paper.

Next, we define the receiver's choice correspondence. We assume that the receiver is a standard Bayesian information processor. That is, upon receiving a signal, she updates her prior belief via Bayes' rule and then chooses from the menu an act that maximizes her expected utility with respect to her posterior belief. Such a standard receiver can be formalized as follows.

\begin{definition}
    [\cite{jakobsen2021axiomatic}] A \textbf{Bayesian representation for $c$} is a pair of a non-constant mixture-linear function $u:\Delta(X)\rightarrow\mathbb{R}$ and a full support probability distribution $\mu\in\Delta(\Omega)$ such that for all $A\in\mathcal{A}$ and $s\in S$,
    \begin{equation}
        c^s (A) = \left\{ f \in A ~:~ \forall g\in A,~ \sum_{\omega\in\Omega} u(f_\omega) \mu^s_\omega \geq \sum_{\omega\in\Omega} u(g_\omega) \mu^s_\omega \right\}.
    \end{equation}
\end{definition}

The above definition is introduced by \cite{jakobsen2021axiomatic}. According to this definition, the receiver is endowed with a von Neumann–Morgenstern utility function $u$ and a prior belief $\mu$, which she uses to process signals and chooses an optimal act. An axiomatic characterization of the Bayesian representation is provided by \cite{jakobsen2021axiomatic}, to which we will return in a later section. Hereafter, we assume that a receiver's choice correspondence $c$ is represented by a Bayesian representation. 

\begin{remark}
    The authors are aware that the assumption $\mathcal{L}(o)=\{o\}$ may not fit some real-life contexts, but the following interpretations are possible.
    First, recall the example of the behavioral interpretation, where the consumer commits to an information acquisition plan by customizing news feeds on social media. 
    In this case, if we read choosing $o$ as terminating the use of social media, the consumer choosing $o$ is free from further recommendation from social media, which is consistent with $\mathcal{L}(o)=\{o\}$.
    
    Next, recall the example of the persuasion interpretation, where the platform provides consumer data to the seller. 
    In this case, the assumption $\mathcal{L}(o)=\{o\}$ is valid when the seller's independent research generates a useful insight only when combined with the data provided by the platform. Indeed, the literature has reported that a seller's product assortment based on independent research becomes clearly more efficient when the platform provides the market data \citep{chen2026role}, which may enhance the validity of the complementarity of the seller's independent research and the platform's data provision discussed here.
\end{remark}

\begin{remark}
    Some reader may be interested in whether the robust Bayesian persuasion representation has an optimal information structure. 
    Let $\sigma^{full}$ be a fully informative information structure, which is formally the $|\Omega|$-dimensional identity matrix $E_{|\Omega|}$.
    Applying Lemma 1 of \cite{dworczak2022preparing}, we can immediately see that, if $\sigma^{full}\in\mathcal{L}(\sigma)$ for every $\sigma\in\mathcal{E}$, then $\sigma^{full}$ maximizes $V(A,\sigma)$ for every menu $A$.
    Intuitively, on the one hand, if an information structure $\sigma$ leads to a better outcome than $\sigma^{full}$ for the sender, the pessimistic sender anticipates that the information structure weakly worse than $\sigma^{full}$ will be chosen from $\mathcal{L}(\sigma)$.
    On the other hand, the sender can rule out outcomes worse than the one resulting from $\sigma^{full}$ just by choosing $\sigma^{full}$ since it must be that $\mathcal{L}(\sigma^{full})=\{\sigma^{full}\}$.
\end{remark}

\section{Examples}\label{Sec_Ex}

\subsection{Consumer data disclosure in digital platforms}\label{Ch_privacy}


In this subsection, we provide further discussion on Example 2 in the introduction:
an app market platform and a developer are the principal and the agent.
We fix a menu available to the developer and delve into the platform's informational choice.

The platform chooses the degree of disclosure of consumer data to the developer.
However, the platform cannot control the additional information available to the developer, which generates deeper insights when combined with the provided data.
It is impossible for the platform to anticipate the content of the additional information, but the platform can know or guess the upper limit of the consumers' \textit{privacy loss}:
it can be interpreted as the legal regulation of the degree of consumer data use for the developer's decision or the developer's capacity limit.

How can privacy loss be formalized? We adopt the concept of ex-post privacy proposed by \cite{eilat2021bayesian}. 
For any probability distributions $p,q\in\Delta(\Omega)$, the KL divergence from $p$ to $q$ is defined as 
\begin{equation}
    D_{KL}(p||q)=\sum_{\omega\in\Omega} p_\omega \log \left(\frac{p_\omega}{q_\omega}\right).
\end{equation}
This metric measures the distance between two probability distributions. The ex-post privacy notion is defined by the signal realization that maximizes the KL divergence, meaning:
\begin{equation}
    \max_{s\in S}D_{KL}(\nu^s||\nu),
\end{equation}
where $\nu\in\Delta(\Omega)$ is a prior.
In other words, privacy loss is determined by the realization under which privacy is most severely compromised.

Now, suppose the information available to the developer is limited by this ex-post privacy notion. Specifically, there exists a constant $\kappa$ such that the developer can only use information structures that satisfy
\begin{equation}
    \max_{s\in S}D_{KL}(\nu^s||\nu) \leq \kappa.
\end{equation}
Assume the platform initially discloses information via an information structure $\sigma$. 
In the example above, the developer could extract additional information about the consumer by combining the disclosed data with other sources. However, in this context, the privacy constraint limits the developer’s ability to do so. Consequently, the privacy loss experienced by the consumer can be formalized as
\begin{equation}
    \mathcal{L}^{\kappa} (\sigma)=\left\{\sigma'\in\mathcal{E}~|~\sigma'\trianglerighteq\sigma \text{ and } \max_{s\in S}D_{KL}(\nu^s||\nu) \leq \kappa\right\}\cup\{\sigma\}.
\end{equation}

To see a more specific situation, suppose that the state space is given by $\Omega=\{\omega_1,\omega_2,\omega_3\}$, which indicates the (representative) consumer's preference.
Let $\nu\in\Delta(\Omega)$ be a common prior such that $\nu_1=\nu_2$.
Moreover, fix the menu $A$, and suppose that the platform prefers $f\in A$ over other acts in $A$ and $c^s(A)=\{f\}$ if $\nu^s_3=\nu_3$.
Finally, suppose $\kappa=(1-\nu_3)\ln 2$.

This setting can be interpreted, for instance, as follows.
The probability of the state $\omega_3$ conveys a lot of information about the consumer's willingness to pay, 
and precise information about this state incentivizes the developer to provide apps that harm consumer welfare through price discrimination in in-app purchases.
If the developer has no information about the state $\omega_3$, then the developer chooses an act $f$, which is the least lucrative app for the developer and the least harmful for the consumer.
The platform's objective here is supposed to be the  protection of consumer welfare, so that $f$ is the most preferable act for the platform.

Then, the information structure $\sigma^*$ given by
\[
\sigma^*=\begin{pmatrix}
    s_1&t_1\\
    s_2&t_2\\
    s_3&t_3
\end{pmatrix}
=\begin{pmatrix}
    1&0\\
    0&1\\
    \frac{1-\nu_2-\nu_3}{1-\nu_3}&\frac{1-\nu_1-\nu_3}{1-\nu_3}
\end{pmatrix}
\]
is optimal for the platform.
To see this, note that, given this information structure, the posteriors are calculated as $(\nu^s_1,\nu^s_2,\nu^s_3)=(1-\nu_3,0,\nu_3)$ and $(\nu^t_1,\nu^t_2,\nu^t_3)=(0,1-\nu_3,\nu_3)$.
That is, this experiment conveys information about the states $\omega_1$ and $\omega_2$, but no information at all about $\omega_3$, which induces the act $f$, the most preferable act for the platform.
Moreover, it follows from $\kappa=(1-\nu_3)\ln 2$ that $\max_{s'\in\sigma^*}D_{KL}(\nu^{s'}||\nu)=\kappa$.
Thus, once the platform chooses $\sigma^*$, the developer cannot acquire additional information anymore, and the most preferable act for the platform is chosen with probability one.

This result implies that data disclosure may be better for consumers because it can consume the developer's capacity, which precludes further acquisition of sensitive information.
Moreover, under another interpretation in which the consumer itself is the principal, this result may contribute to the understanding of the empirically reported behavior called the privacy paradox: consumers with privacy concerns easily and voluntarily provide their data \citep[see, e.g.,][]{acquisti2016economics}.

\subsection{Fact-checking}
\label{sec_fact}
In the standard Bayesian persuasion problem, the sender often withholds some information and avoids full disclosure. Such analyses typically assume that the receiver does not engage in additional verification or fact-checking of the information received. In reality, however, the receiver may seek to verify the accuracy of the information by conducting additional fact-checking or gathering further evidence. A sender who is concerned about such behavior can be captured in our framework.
Here, signals are indexed by states, with signal $\omega'$ interpreted as a report that the state is $\omega'$. We write $\sigma(\omega' | \omega)$ for the probability of this report conditional on the true state being $\omega$.
For a given $P\subseteq[0,1]$, 

\begin{equation}
    \mathcal{L}^P(\sigma)=\left\{\sigma'\in\mathcal{E}~|~\forall\omega\in\Omega,~\sigma'(\omega|\omega)=\sigma(\omega|\omega)+p\sum_{\omega'\in\Omega\setminus\{\omega\}}\sigma(\omega'|\omega),~p\in P\right\}.
\end{equation}
It is worth noting that, if the state is binary and $P$ is a singleton, the model reduces to \cite{ederer2022bayesian}, where the true state is revealed to the receiver with a certain probability.

Notice that for each $\omega\in\Omega$ and $p\in P$, 
\begin{equation}
    \sigma'(\omega|\omega)
    =\sigma(\omega|\omega)+p\sum_{\omega'\in\Omega\setminus\{\omega\}}\sigma(\omega'|\omega)
    =(1-p)\sigma(\omega|\omega) + p.
\end{equation}
Thus, the information structure $\sigma'$ is represented by a mixed matrix $(1-p)\sigma\cup p\sigma^{full}$, where $\sigma^{full}$ is a fully informative information structure. Then, for each $A\in\mathcal{A}$, 
\begin{align}
    V^A(\sigma)
    &=\inf_{\pi\in \mathcal{L}^P(\sigma)} \sum_{s \in \pi} \left[ \min_{f^s \in c^s (A)} \sum_{\omega \in \Omega} v (f_\omega) s_\omega \nu_\omega \right] \\
    &=\inf_{p\in P}\left[(1-p) \sum_{s \in \sigma} \left[ \min_{f^s \in c^s (A)} \sum_{\omega \in \Omega} v (f_\omega) s_\omega \nu_\omega \right] +p V^A \left (\sigma^{full}\right)\right].
\end{align}
When the sender is concerned about fact-checking, it is reasonable to assume that full disclosure is the least preferred. In this case, for any $p$, the first term in the expression above is larger than the second term, implying that the same $p \in P$ is used for every information structure $\sigma$. Therefore, as long as we fix the menu and observe only the sender's preference over information structures, we cannot distinguish whether the sender is concerned about fact-checking or not. However, once we extend the analysis to our framework, these cases can be distinguished by comparing the corresponding lottery equivalents.

\subsection{Temptation}
Consider a voter who chooses a politician in an election. Suppose the voter intends to select a candidate solely on the candidate's fundamental competence as a politician. At this stage, the voter considers only information about these intrinsically relevant abilities necessary. However, as the election approaches, the voter may begin to weigh factors unrelated to competence  such as appearance or scandals. This leads her to seek additional information about these aspects. As a result, the future self may end up choosing a candidate whom the present self would not have preferred.

Our framework can also capture this type of DMs concerned about their future selves' information acquisition. In this case, the robust Bayesian persuasion representation can be interpreted as an intrapersonal game between the present self and the future self, with the former acting as the sender and the latter as the receiver.

To see this with a concrete example, let $\Omega=\{Gg, Gb, Bg, Bb\}$. The current self is only concerned with whether the realized state is $G$ or $B$, and is almost indifferent to whether it is $g$ or $b$. In contrast, the future self cares not only about whether the state is $G$ or $B$, but also about whether it is $g$ or $b$, and makes her choice accordingly. For some $\varepsilon^*,\varepsilon\in\mathbb{R}_{++}$ such that $\varepsilon^*>\varepsilon\approx 0$, this difference in preferences between the current and future selves can be formalized as a difference in the evaluation of an act $f$, as follows: 
\begin{equation}
    \text{$v(f(Gg))=1+\varepsilon$, $v(f(Gb))=1-\varepsilon$, $v(f(Bg))=\varepsilon$, $v(f(Bb))=-\varepsilon$}
\end{equation}
and
\begin{equation}
    \text{$u(f(Gg))=1+\varepsilon^*$, $u(f(Gb))=1-\varepsilon^*$, $u(f(Bg))=\varepsilon^*$, $u(f(Bb))=-\varepsilon^*$}.
\end{equation}
Here, suppose that the current self chooses an information structure that identifies only whether the state is $G$ or $B$, in order to control future information acquisition. Such an information structure is given by
\begin{equation}
    \sigma =\begin{pmatrix}
        1 &1 & 0 & 0 \\
        0 &0 &1 & 1
    \end{pmatrix}^T.
\end{equation}
If the current self's preference is represented by a robust Bayesian persuasion representation, she may be concerned that the future self will acquire additional information that distinguishes between $g$ and $b$. For instance, the current self may anticipate that, with some probability, the future self reaches full disclosure. In this case, the corresponding set $\mathcal{L}(\sigma)$ when $\sigma$ is chosen is given by
\begin{align}
    \mathcal{L}(\sigma) &=\left\{ \alpha\sigma+(1-\alpha)\sigma^{full} ~\mid~ \alpha\in M \right\} \\
    & =\left\{ \alpha\begin{pmatrix}
        1 &1 & 0 & 0 \\
        0 &0 &1 & 1
    \end{pmatrix}^T \cup (1-\alpha) E_{|\Omega|} ~\mid~ \alpha\in M\right\},
\end{align}
where $M\subset [0,1]$.

Since the current self evaluates the future self's information acquisition pessimistically, she may deliberately choose an information structure that also reveals information about $g$ and $b$ at the stage of selecting the information structure. In particular, \cite{kitano2026acquiring} show that when $\mathcal{L}(\sigma)$ contains all information structures that are Blackwell more informative than $\sigma$ for any $\sigma$, it can be optimal for the current self to deliberately provide the future self with information about $g$ and $b$. Whether such behavior---namely, the intentional acquisition of otherwise irrelevant information---arises in more general environments remains an open question.

\section{Axiomatic characterization}

In this section, we provide an axiomatic characterization of the robust Bayesian persuasion representation. We first analyze the sender's preference over pairs consisting of a menu and an information structure. We then turn to the receiver's preference over acts conditional on signals. Our main contribution lies in the analysis of the sender's preference, while the characterization of the receiver's preference is the result of \cite{jakobsen2021axiomatic}. 

\subsection{Sender's preference}

The sender we consider in this paper is concerned that, after choosing a particular information structure, a more informative information structure that is unfavorable to her may be implemented. In this sense, our sender differs from the standard sender in the Bayesian persuasion framework. Nevertheless, we assume that she satisfies basic rationality axioms over her preference. The following axiom captures this feature.

\begin{axiom}{\textbf{Basic rationality.}}
    $\succsim$ is complete and transitive. Moreover, for every $\sigma\in\mathcal{E}$ and $A,B,C\in\mathcal{A}$, the sets 
    \begin{equation}
       \{\alpha\in[0,1]~|~(\alpha A+(1-\alpha)B,\sigma)\succsim(C,\sigma)\}
    \end{equation}
    and 
    \begin{equation}
       \{\alpha\in[0,1]~|~(A,\sigma)\succsim(\alpha B+(1-\alpha)C,\sigma)\}
    \end{equation}
    are closed.
\end{axiom}

The first condition is standard. The second condition requires that, for each fixed information structure, the sender's conditional preference over menus satisfies \textit{mixture continuity}. Note that we do not assume \textit{mixture continuity} over pairs of menus and information structures. Although one could impose such a stronger continuity condition, the present condition is sufficient for our purposes.

Recall that the information structure $o$ represents the special option in which the sender provides no information to the receiver. The following axiom requires that, given this information structure, the sender's preference over acts is represented by a subjective expected utility model.

\begin{axiom}{\textbf{SEU with no information.}}
\label{axiom_SEU}
    $\succsim$ conditional on $\mathcal{F}\times\{o\}$ satisfies \cite{anscombe1963definition}'s axioms. That is, weak order, non-triviality, mixture continuity, monotonicity, and independence.
\end{axiom}

Notice that when the sender selects a singleton menu the receiver has no choice but to implement that act. Therefore, the sender's conditional preference in this case coincides with her underlying preference over acts. \textit{SEU with no information} thus requires that this preference be represented by a standard SEU under no information.

In fact, one could assume that preferences over acts conditional on any information structure are represented by SEU without affecting the characterization result. However, as we show in the proof, it is sufficient to impose this condition only for pairs consisting of an act and the null information structure $o$.

The next axiom concerns the sender's preference over lotteries. Since the value of a lottery is independent of the information structure, the sender should evaluate a given lottery in the same way regardless of which information structure it is paired with. The following axiom formalizes this intuition.

\begin{axiom}{\textbf{Lottery invariance.}}
\label{axiom_LI}
    For every information structure $\sigma$ and lottery $l$, $(\{l\},\sigma)\sim(\{l\},o)$.
\end{axiom}

The statement of this axiom itself imposes a weaker requirement than the preceding discussion. In particular, \textit{lottery invariance} only requires that $(\{l\}, \sigma)$ be indifferent to $(\{l\}, o)$. However, since $\succsim$ satisfies \textit{transitivity}, it immediately implies that, for any $\sigma$ and $\sigma'$, $(\{l\},\sigma)\sim(\{l\},\sigma')$ holds.

So far, we have introduced axioms that are independent of the sender's concern about information leakage. The axioms that follow are primarily concerned with the sender's attitude toward information leakage. Before introducing the remaining axioms, we need to define several new objects. Fix an arbitrary information structure $\sigma$, and consider a menu indexed by the signals that constitute $\sigma$, $\{f^s\}_{s\in\sigma}$.
This signal-indexed collection of acts can be interpreted as a contingent plan that specifies which act to choose upon observing each signal.

Using this notation, we define the act induced by $\sigma$ and $\{f^s\}_{s\in\sigma}$.

\begin{definition}
    For an information structure $\sigma$ and a menu $\{f^s\}_{s\in\sigma}$, \textit{the act induced by $\sigma$ and $\{f^s\}_{s\in\sigma}$} is defined as
    \begin{equation}
        \sum_{s\in\sigma}sf^s := \left(\sum_{s\in\sigma}s_\omega f^s_\omega\right)_{\omega\in\Omega}.
    \end{equation}
\end{definition}

Under this definition, the act induced by $\sigma$ and $\{f^s\}_{s \in \sigma}$ yields, at each state, a lottery defined as a convex combination determined by $\sigma$. The intuition is as follows: First, a state is realized, and conditional on that state, a signal is randomly generated according to $\sigma$. Given the realized signal, the act indexed by that signal is selected from the menu $\{f^s\}_{s \in \sigma}$. Since the signal generated conditional on a given state is generally random, taking expectations over this randomness yields $\sum_{s\in\sigma}s_\omega f^s_\omega$. The act induced by $\sigma$ and $\{f^s\}_{s \in \sigma}$ is obtained by reducing this compound lottery after the realization of the state. Similar objects are used in, for example,  \cite{jakobsen2021axiomatic} and \cite{wang2024regret}.
See Figure \ref{fig:induced_act} for a graphical explanation.

\begin{figure}
    \centering
    \includegraphics[width=0.75\linewidth]{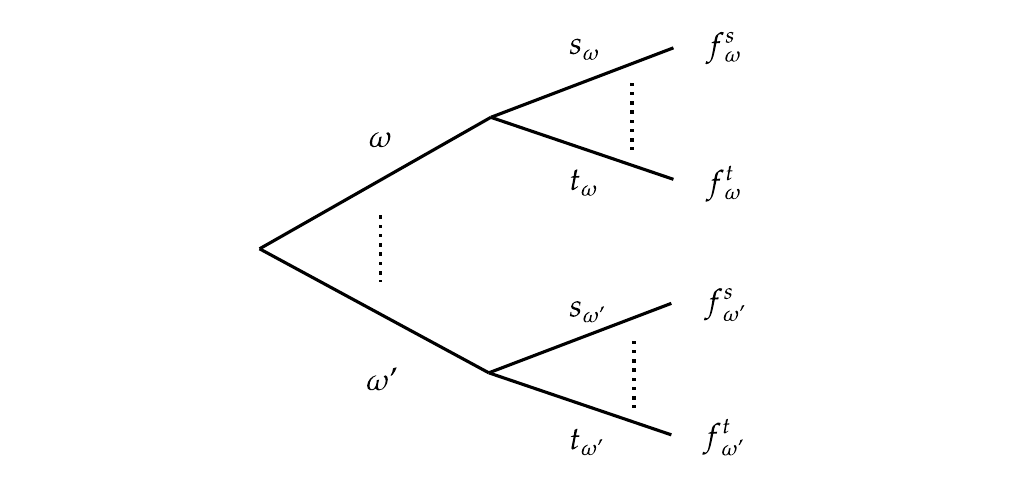}
    \caption{\textbf{Illustration of an induced act.} First, a state is realized, and a signal is then generated conditional on that state. The signal determines which act is selected, and the selected act yields the lottery corresponding to the realized state.}
    \label{fig:induced_act}
\end{figure}

Together with \textit{SEU with no information}, choosing these acts under no information can be interpreted as a commitment to implement a plan $\{f^s\}_{s\in\sigma}$ that is contingent on the signal generated by the information structure $\sigma$. Indeed, by \textit{SEU with no information}, there exist a von Neumann–Morgenstern utility function $v:\Delta(X)\rightarrow\mathbb{R}$ and a prior $\nu\in\Delta(\Omega)$ such that $\left( \sum_{s\in\sigma}sf^s, o \right)$ is evaluated as follows: 

\begin{equation}
    V \left( \sum_{s\in\sigma}sf^s, o \right) =\sum_{\omega\in\Omega} v\left( \left[\sum_{s\in\sigma}sf^s\right]_\omega \right) \nu_\omega =\sum_{s\in\sigma,\omega\in\Omega} \left[ \sum_{\omega'\in\Omega} v(f^s_{\omega'}) \nu^s_{\omega'} \right]s_\omega\nu_\omega.
\end{equation}
The term inside the bracket in the final expression represents the expected utility of $f^s$, computed using the Bayesian posterior of $\nu$ given the signal $s$. The value of the induced act is given by the expectation of these terms.

Using these objects, we consider the following set:
\begin{equation}
    c^{\sigma}(A):=\left\{\sum_{s\in\sigma}sf^s~|~ f^s\in c^s(A)\right\}.
\end{equation}
This set consists of induced acts formed by the receiver's optimal choices for each signal. From the discussion above, this set can be interpreted as the set of possible benefits that the sender can obtain from the menu $A$, given that the information structure $\sigma$ generates signals.

Under this interpretation of $c^{\sigma}(A)$, we require that the sender anticipates that tie-breaking by the receiver is resolved in the most unfavorable way for her. This condition is formalized by imposing a \textit{preference for commitment} on the sender's preference over $c^{\sigma}(A)$ and its elements.

\begin{axiom}{\textbf{Sender pessimism.}}
\label{axiom_SP}
    For every information structure $\sigma$ and menu $A$, (i) $(\{f\},o)\succsim(c^\sigma(A),o)$ for all $f\in c^\sigma(A)$; (ii) $(\{g\},o)\sim(c^\sigma(A),o)$ for some $g\in c^\sigma(A)$.
\end{axiom}

The intuition behind \textit{sender pessimism} can be decomposed into two components. The first condition states that, when the receiver chooses from a menu based on $\sigma$ and multiple plans are optimal, the plan that is less favorable to the sender is selected. Hence, fixing a plan ex ante is preferred by the sender to delegating tie-breaking to the receiver.
The second condition is more straightforward: even if $c^{\sigma}(A)$ consists of multiple elements, one plan must ultimately be selected from this set.

The following axiom states that the value of a pair $(A, \sigma)$ is bounded by the evaluations of commitments induced by information structures that are more informative than $\sigma$.

\begin{axiom}{\textbf{Betweenness.}}
    For all $A\in\mathcal{A}$ and $\sigma\in\mathcal{E}$, $\left(c^{\sigma'}(A),o\right)\succsim(A,\sigma)\succsim\left(c^{\sigma''}(A),o\right)$ for some $\sigma',\sigma''\in\mathcal{E}$ with $\sigma',\sigma''\trianglerighteq\sigma$.
\end{axiom}

The intuition for this axiom can be understood by considering cases in which it is violated. Suppose first that, for all $\sigma' \trianglerighteq \sigma$, the commitment $\left(c^{\sigma'}(A), o\right)$ is strictly preferred to $(A, \sigma)$. In this case, the sender perceives a loss from $(A, \sigma)$ that is more severe than any loss arising from information leakage. Conversely, suppose that $(A, \sigma)$ is strictly preferred to $\left(c^{\sigma'}(A), o\right)$ for all $\sigma' \trianglerighteq \sigma$. Then the sender does not care at all about the payoff loss resulting from information leakage. In summary, \textit{betweenness} requires that the sender is concerned about information leakage, but does not perceive any other source of payoff loss that is more severe than that arising from information leakage.

To state the last axiom, we introduce a specific set of menus.

\begin{definition}
    If a menu contains at least two elements and at least one non-constant act, we call such a menu \textit{essential}. We denote by $\mathcal{A}^*$ the set of all essential menus.
\end{definition}

Note that if a menu consists only of lotteries or a single element, information structures have no effect on the decisions of either the sender or the receiver. This is because the evaluation of lotteries does not depend on states, and in the case of singleton menus, there is no scope for utilizing signals. We therefore refer to the class of menus for which information structures can have an effect as essential menus.

\begin{axiom}{\textbf{Menu-independent leakage}.}
    For $A,B\in\mathcal{A}^*$ and $\sigma\in\mathcal{E}$, if there exists $\sigma^\ast\in\mathcal{E}$ such that $\sigma^\ast\trianglerighteq\sigma$ and $(A,\sigma)\sim(c^{\sigma^\ast}(A),o)$, then $(c^{\sigma^\ast}(B),o)\succsim(B,\sigma)$.
\end{axiom}

The intuition behind \textit{menu-independent leakage} is as follows. Consider a pair consisting of a menu $A$ and an information structure $\sigma$, and suppose that there exists $\sigma^\ast \in \mathcal{E}$ such that $\sigma^\ast \trianglerighteq \sigma$ and $(A, \sigma) \sim (c^{\sigma^\ast}(A), o)$. In this case, the sender can be interpreted as believing that, when she chooses $\sigma$, information leakage occurs so that the receiver effectively obtains information equivalent to that generated by $\sigma^\ast$.

\textit{Menu-independent leakage} requires that, in such a situation, the sender takes into account that $\sigma$ may effectively be transformed into $\sigma^\ast$ not only for the menu $A$, but also for any other menu $B$. In other words, the sender's expectation about how information leaks from a given information structure $\sigma$ should not depend on the menu she chooses. In this sense, the axiom requires that concerns about information leakage be independent of the chosen menu.

The following theorem states that any preference satisfying the above axioms admits a robust Bayesian persuasion representation.

\begin{theorem}
\label{theorem_RBP}
    If $\succsim$ satisfies \textit{basic rationality}, \textit{SEU with no information}, \textit{lottery invariance}, \textit{sender pessimism}, \textit{betweenness}, and \textit{menu-independent leakage}, then it is represented by a robust Bayesian persuasion representation.
\end{theorem}

The formal proof of Theorem \ref{theorem_RBP} is contained in \hyperref[Appendix]{Appendix}. Note that this theorem does not provide an if-and-only-if characterization. The key parameter in the robust Bayesian persuasion representation is $\mathcal{L}$, but the original definition imposes relatively weak restrictions on this parameter. In particular, given $\sigma\in\mathcal{E}$, $\mathcal{L}(\sigma)$ may fail to include a sufficiently rich set of information structures for the representation to satisfy \textit{menu-independent leakage}. To address this issue, we focus on a particular subclass of robust Bayesian persuasion representations.

\begin{definition}
    $\mathcal{L}:\mathcal{E}\rightarrow 2^{\mathcal{E}}$ is \textit{inclusive} if for all $\sigma\in\mathcal{E}$,
    \begin{equation}
        \text{there exist $A\in\mathcal{A}^*$ and $\sigma^*\trianglerighteq\sigma$ such that $(A,\sigma)\sim(c^{\sigma^\ast}(A),o)$}\implies\sigma^*\in\mathcal{L}(\sigma).
    \end{equation}
\end{definition}

An inclusive $\mathcal{L}$ contains all information leakages that the sender takes into account when evaluating essential menus.
Restricting the condition to essential menus $\mathcal{A}^*$ is important. Indeed, if we impose the above condition for all menus $A$ and $B$, then $(\{l\},\sigma)\sim(c^{\sigma'}(\{l\}),o)$ holds for every lottery $l$ and $\sigma'$ with $\sigma'\trianglerighteq\sigma$, which would imply $\mathcal{L}(\sigma)=\{\sigma'\in\mathcal{E}~\mid~\sigma'\trianglerighteq\sigma\}$. As a result, the set $\mathcal{L}(\sigma)$ necessarily coincides with the entire set of feasible information structures. Our condition is necessary to capture a wide range of the sender's concerns about information leakage. 

The following results state that the axioms are necessary conditions for preferences to admit a robust Bayesian persuasion representation with an inclusive $\mathcal{L}(\sigma)$.

\begin{proposition}
\label{proposition_RBP_nec}
    If $(\succsim,c)$ is represented by a robust Bayesian persuasion representation and $\mathcal{L}$ is inclusive, then it satisfies all of the axioms.
\end{proposition}

The proof of Proposition \ref{proposition_RBP_nec} is contained in \hyperref[Appendix]{Appendix}. Combining Proposition \ref{proposition_RBP_nec} and Theorem \ref{theorem_RBP}, we obtain the following characterization result as a corollary.
    
\begin{corollary}
    $(\succsim,c)$ satisfies all of the axioms if and only if it is represented by a robust Bayesian persuasion representation and $\mathcal{L}$ is inclusive.
\end{corollary}

\subsection{Aversion to information leakage}

The parameter $\mathcal{L}$, which captures the sender's subjective perception of information leakage, raises the question of how well it can be identified from observed choice behavior. More precisely, to what extent is the parameter $\mathcal{L}$ uniquely determined? Unfortunately, the uniqueness cannot be guaranteed even when $\mathcal{L}$ is inclusive. The reason is that one can enlarge an inclusive $\mathcal{L}(\sigma)$ by adding information structures that are never relevant for the evaluation of every menu, without affecting the sender's choice behavior. We therefore define below the largest $\mathcal{L}$ that contains all such irrelevant options.

\begin{definition}
    Given a set of information structures $\mathcal{L},$ a mapping $\bar{\mathcal{L}}:\mathcal{E}\rightarrow 2^{\mathcal{E}}$ is \textit{maximal} if for all $\sigma\in\mathcal{E}\setminus \{o\}$, 
    \begin{equation}
        \bar{\mathcal{L}}(\sigma)=\mathcal{L}(\sigma)\bigcup \left[\bigcap_{A\in\mathcal{A}^*}\left\{ \sigma'\trianglerighteq\sigma~|~ V\left(c^{\sigma'}(A),o\right)\geq V(A,\sigma) \right\}\right].
    \end{equation}
\end{definition}

The maximal $\bar{\mathcal{L}}$ contains both the information leakage that the sender takes into account when evaluating some menu and the irrelevant information leakage that is never considered in the evaluation of any essential menu. Indeed, the former satisfies $V\left(c^{\sigma'}(A),o\right)= V(A,\sigma)$, while the latter satisfies $V\left(c^{\sigma'}(A),o\right)> V(A,\sigma)$. This interpretation is also supported by the fact that the maximal $\bar{\mathcal{L}}$ can be rewritten as
\begin{equation}
    \bar{\mathcal{L}}(\sigma)=\mathcal{L}(\sigma)\bigcup \left[\bigcap_{A\in\mathcal{A}^*}\left\{ \sigma'\trianglerighteq\sigma~|~ \left(c^{\sigma'}(A),o\right)\succsim (A,\sigma) \right\}\right]
\end{equation}
for each $\sigma\in\mathcal{E}\setminus\{o\}$.

Moreover, whenever a robust Bayesian persuasion representation $(v,\nu,\mathcal{L})$ represents $(\succsim,c)$, the maximal $\bar{\mathcal{L}}$ constructed from it necessarily contains the original $\mathcal{L}$. In this sense, the maximal $\bar{\mathcal{L}}$ is the mapping that assigns, within the robust Bayesian persuasion representation of the sender's preference, the largest set to each information structure. Therefore, adding any information structure to the maximal $\bar{\mathcal{L}}$ would change the sender's preference. Indeed, if $\sigma'\notin\bar{\mathcal{L}}(\sigma)$ for some $\sigma$ and $\sigma'$, we can find some essential menu $A$ such that $(A,\sigma)\succ\left(c^{\sigma'}(A),o\right) $. Thus, adding this information structure would change the worst-case information leakage for the sender evaluating the pair of $A$ and $\sigma$. 

By taking the maximal $\bar{\mathcal{L}}$, we obtain uniqueness of the parameters.

\begin{proposition}
\label{proposition_uniqueness}
    If $(v,\nu,\mathcal{L})$ and $(v',\nu',\mathcal{L}')$ represent $(\succsim,c)$, then there exist $\alpha>0$ and $\beta\in\mathbb{R}$ such that $v'=\alpha v+\beta$, $\nu'=\nu$, and $\bar{\mathcal{L}}=\bar{\mathcal{L}}'$.
\end{proposition}

The proof of Proposition \ref{proposition_uniqueness} is contained in \hyperref[Appendix]{Appendix}. Having established the uniqueness of the parameter via the maximal $\bar{\mathcal{L}}$, we use it to conduct comparative statics on the sender's concern about information leakage. Intuitively, the larger $\bar{\mathcal{L}}$ is, the lower the value of an information structure should be. The following result confirms this intuition.

\begin{proposition}
\label{proposition_comparative}
    Suppose that $i\in\{1,2\}$ and $(\succsim_i, c)$ is represented by a robust Bayesian persuasion representation $(v,\nu,\bar{\mathcal{L}}_i)$. Then, the following conditions are equivalent:
    \begin{enumerate}
        \item For all $A\in\mathcal{A}$, $l\in\Delta(X)$, and $\sigma\in\mathcal{E}$,
        \begin{equation}
            (\{l\},o)\succsim_1 (A,\sigma) \implies (\{l\},o)\succsim_2 (A,\sigma)
        \end{equation}
        \item $\bar{\mathcal{L}}_1(\sigma)\subseteq \bar{\mathcal{L}}_2(\sigma)$ for all $\sigma\in\mathcal{E}$.
    \end{enumerate}
\end{proposition}

The proof of Proposition \ref{proposition_comparative} is contained in \hyperref[Appendix]{Appendix}. The second condition states that, for any information structure, sender 2 is always more concerned about information leakage than sender 1. The first condition states that if sender 1 chooses not to provide information due to concerns about information leakage, then sender 2 does so as well. Proposition \ref{proposition_comparative} establishes that these two statements are equivalent. This result also implies that the maximal $\bar{\mathcal{L}}_i$ is useful for comparing the extent of information avoidance across senders.

\begin{corollary}
\label{corollary_information_avoidance}
    Suppose that $i\in\{1,2\}$ and $(\succsim_i, c)$ is represented by a robust Bayesian persuasion representation $(v,\nu,\bar{\mathcal{L}}_i)$. Moreover, suppose that $\bar{\mathcal{L}}_1(\sigma)\subseteq \bar{\mathcal{L}}_2(\sigma)$ holds for all $\sigma\in\mathcal{E}$. Then, for any $A\in\mathcal{A}$ and $\sigma\in\mathcal{E}$, 
    \begin{equation}
        (A,o)\succsim_1 (A,\sigma) \implies (A,o)\succsim_2 (A,\sigma).
    \end{equation}
\end{corollary}

In other words, a sender with a larger maximal $\bar{\mathcal{L}}_i$ is more likely to refrain from providing information. The proof of Corollary \ref{corollary_information_avoidance} is contained in \hyperref[Appendix]{Appendix}.

\section{Discussion}

In this section, we discuss several topics related to the robust Bayesian persuasion representation. First, we examine the receiver's choice correspondence. Next, we clarify the differences between our model and the standard Bayesian persuasion model from an axiomatic perspective.

\subsection{Receiver's choice}

So far, we have proceeded under the assumption that the receiver's choice can be represented as a Bayesian information processor. However, it is not immediately clear which behavioral principles justify such an assumption. \cite{jakobsen2021axiomatic} shows that the following axioms are equivalent to the receiver's choice correspondence admitting a Bayesian representation.

\begin{axiom}{\textbf{Standard receiver preferences.}}
\label{axiom_EU}
    The following properties hold:
    \begin{enumerate}
        \item[(i)] If $f,g\in A\cap B$, $f\in c^s(A)$, and $g\in c^s(B)$, then $f\in c^s(B)$.
        \item[(ii)] For every $s\in S$, there exists $A$ such that $c^s(A)\neq A$.
        \item[(iii)] For all $A,B\in\mathcal{A}$, $\alpha\in[0,1]$, and $s\in S$, $c^s(\alpha A+(1-\alpha)B) \subseteq \alpha c^s(A)+(1-\alpha) c^s(B)$.
        \item[(vi)] For every $s\in S$, $c^s$ is upper hemicontinuous.
        \item[(v)] For every finite set $L\subseteq \Delta(X)$, $h,h'\in\mathcal{F}$, and $s_\omega,s_{\omega'}>0$, $c^s_\omega(L[\omega]h)\subseteq c^{s'}_{\omega'}(L[\omega']h')$. 
    \end{enumerate}
\end{axiom}

\begin{axiom}{\textit{Bayesian consistency.}} 
\label{axiom_BC}
    If $tf+(1-t)h\in c^s(tA+(1-t)h)$, then $sf+(1-s)h'\in c^t(sA+(1-s)h')$.
\end{axiom}

\begin{theorem}{(Theorem 2 in \cite{jakobsen2021axiomatic})}
    $c$ satisfies \textit{standard receiver preferences} and \textit{Bayesian consistency} if and only if it admits a Bayesian representation.
\end{theorem}

We conclude this subsection with one remark on the receiver's choice correspondence. The assumption that the receiver is a Bayesian information processor is essential for our results. In particular, note that Lemma \ref{lemma_final} in the proof of Theorem \ref{theorem_RBP} relies on the choice correspondence being represented by a Bayesian representation. Therefore, our characterization does not extend to non-Bayesian receivers. Extending the analysis in this direction is left for future research.

\subsection{Bayesian persuasion representation}

The robust Bayesian persuasion representation differs from the standard Bayesian persuasion model in two key respects. First, the sender represented by the robust Bayesian persuasion model is concerned that an information structure more informative than the one she chooses may be implemented. In contrast, the sender in the standard Bayesian persuasion has no such concern and expects that the chosen information structure will be implemented as intended. Second, while the robust Bayesian persuasion representation adopts sender pessimism in resolving tie-breaking by the receiver, the standard Bayesian persuasion model adopts sender optimism. Under these differences, the Bayesian persuasion representation is defined as follows.

\begin{definition}
    A \textbf{Bayesian persuasion representation for $(\succsim,c)$} is a function 
    \begin{equation}
        W(A,\sigma)=\sum_{s\in\sigma,\omega\in\Omega} \left[ \max_{f\in c^s(A)} \sum_{\omega\in\Omega} v(f_\omega) \nu^s_\omega \right]s_\omega \nu_\omega,
    \end{equation}
    where $v:\Delta(X)\rightarrow\mathbb{R}$ is a non-constant and mixture-linear function, and $\nu\in\Delta(\Omega)$ has full support.
\end{definition}

Recall that $c^{\sigma}(A)$ can be interpreted as the sender's ex-ante value of the menu $A$ when the information structure $\sigma$ is implemented with certainty. In this case, it coincides with the value of the pair $(A, \sigma)$ under the Bayesian persuasion representation. Therefore, the value of $(A, \sigma)$ should be equal to that of $(c^{\sigma}(A), o)$. The following axiom formalizes this property.

\begin{axiom}{\textbf{Reduction.}}
    For every menus $A,B$, and information structure $\sigma$, $(c^{\sigma}(A),o)\succsim (c^{\sigma}(B),o)$ if and only if $(A,\sigma)\succsim (B,\sigma)$.
\end{axiom}

The second axiom is the counterpart of \textit{sender pessimism}. Its intuition is identical to that of \textit{sender pessimism}, except that tie-breaking is resolved in a way favorable to the sender.

\begin{axiom}{\textbf{Sender optimism.}}
\label{axiom_SO}
    For every information structure $\sigma$ and menu $A$, (i) $(c^\sigma(A),o)\succsim(\{f\},o)$ for all $f\in c^\sigma(A)$; (ii) $(\{g\},o)\sim(c^\sigma(A),o)$ for some $g\in c^\sigma(A)$.
\end{axiom}

Together with \textit{basic rationality}, \textit{SEU with no information}, and \textit{lottery invariance}, we obtain the following result.

\begin{corollary}
\label{corollary_BP}
    $\succsim$ satisfies \textit{basic rationality}, \textit{SEU with no information}, \textit{lottery invariance}, \textit{reduction}, and \textit{sender optimism} if and only if it is a Bayesian persuasion representation.
\end{corollary}

The proof of Corollary \ref{corollary_BP} is contained in \hyperref[Appendix]{Appendix}. We remark that replacing \textit{sender pessimism} with \textit{sender optimism }immediately yields a robust Bayesian persuasion representation in which tie-breaking is resolved optimistically. This observation shows that the two minimization operators appearing in the robust Bayesian persuasion representation are derived by two independent axiomatic conditions.

\subsection{Alternative representations}

\subsubsection{Fact-checking representation}

As discussed in Section \ref{sec_fact}, the fact-checking model provides an important example for identifying whether the sender is concerned about information leakage. In this subsection, we provide an axiomatic foundation for this special case. We first define the fact-checking representation.

\begin{definition}
    A \textbf{fact-checking representation for $(\succsim,c)$} is a function 
    \begin{equation}
        \hat V(A,\sigma)
        =
        \inf_{p\in P(\sigma)}\left[(1-p) \sum_{s \in \sigma} \left[ \min_{f^s \in c^s (A)} \sum_{\omega \in \Omega} v (f_\omega) s_\omega \nu_\omega \right] +p \hat V \left (A,\sigma^{full}\right)\right]
    \end{equation}
    where $v:\Delta(X)\rightarrow\mathbb{R}$ is a non-constant and mixture-linear function, $\nu\in\Delta(\Omega)$ has full support, $P(\sigma)\subseteq [0,1]$, and $\sigma^{full}$ is the $|\Omega| \times |\Omega|$ identity matrix.
\end{definition}

The representation above is obtained by rewriting the model in Section \ref{sec_fact}. More precisely, the set of probabilities considered in the fact-checking representation depends on the information structure $\sigma$. A larger set means that the sender considers a wider range of fact-checking accuracy levels.

We next introduce an axiom that characterizes the fact-checking representation. For $p\in[0,1]$, let $\pi_p(\sigma):=(1-p)\sigma\cup p\sigma^{\mathrm{full}}$. This experiment implements $\sigma$ with probability $1-p$ and fully reveals the state with probability $p$.
The receiver observes which experiment is implemented, as well as the signal it generates.

\begin{axiom}{\textbf{Fact-checking betweenness.}}
    For every menu $A$ and information structure $\sigma$, there exist $p,q\in[0,1]$ such that
    \[
    (c^{\pi_p(\sigma)}(A),o)
    \succsim (A,\sigma)
    \succsim (c^{\pi_q(\sigma)}(A),o).
    \]
\end{axiom}

\textit{Fact-checking betweenness} strengthens \textit{betweenness} by restricting the bounding information structures to mixtures of the chosen information structure and full disclosure.
It requires the sender's evaluation of $(A,\sigma)$ to lie between the evaluations of the induced menus under two possible fact-checking probabilities.
Thus, the sender's evaluation can be bounded by outcomes arising from fact-checking alone.

By replacing \textit{betweenness} with \textit{fact-checking betweenness}, we obtain a fact-checking representation.

\begin{proposition}
\label{proposition_fact-checking}
    If $\succsim$ satisfies \textit{basic rationality}, \textit{SEU with no information}, \textit{lottery invariance}, \textit{sender pessimism}, \textit{menu-independent leakage}, and \textit{fact-checking betweenness}, then it is represented by a fact-checking representation.
\end{proposition}

\subsubsection{Garbling-robust Bayesian persuasion}

So far, we have focused on a sender who is concerned about information leakage. Such a sender anticipates that the information structure implemented may be more informative than the one she originally chooses. In some situations, however, the opposite concern may arise: the implemented information structure may be less informative than the sender's choice. For example, the receiver may fail to notice some information or irrationally discard part of the information (for discussions of such behaviors, see \cite{golman2017information}).
Moreover, a rationally inattentive agent may also ignore some part of information due to attention costs, which is incorporated into the BP model by \cite{lipnowski2020attention}.
A sender with such concerns can be modeled as follows.

\begin{definition}
    A \textbf{garbling-robust Bayesian persuasion representation for $(\succsim,c)$} is a function 
    \begin{equation}
        V^{\dagger}(A,\sigma)=\inf_{\pi\in \mathcal{L}^{\dagger}(\sigma)} \sum_{s\in\pi,\omega\in\Omega} \underline{V} (\nu ^s) s_\omega \nu_\omega,
    \end{equation}
    where
    \begin{equation}
        \underline{V} (\nu ^s)=\min_{f^s \in c^s (A)} \sum_{\omega \in \Omega} v (f_\omega) \nu^s_\omega
    \end{equation}
    is the sender's interim expected utility after the realization of the signal $s$, $v:\Delta(X)\rightarrow\mathbb{R}$ is a non-constant and mixture-linear function, $\nu\in\Delta(\Omega)$ has full support, and $\mathcal{L}^{\dagger}:\mathcal{E}\rightarrow 2^{\mathcal{E}}$ satisfies the following condition: $\mathcal{L}^{\dagger}(\sigma)\subseteq \{\sigma'\in\mathcal{E}~|~\sigma\trianglerighteq\sigma'\}$ for all $\sigma\in\mathcal{E}$.
\end{definition}

Under the garbling-robust Bayesian persuasion representation, the sender anticipates a set $\mathcal{L}^{\dagger}(\sigma)$ consisting of information structures that are Blackwell less informative than $\sigma$ and evaluates menus pessimistically. Unlike in the robust Bayesian persuasion representation, the sender does not consider the possibility that a more informative information structure will be implemented. Observe that $\mathcal{L}^{\dagger}(o)=\{o\}$ follows from the definition of the model.

\begin{axiom}{\textbf{Betweenness$^\dagger$.}}
    For all $A\in\mathcal{A}$ and $\sigma\in\mathcal{E}$, $\left(c^{\sigma'}(A),o\right)\succsim(A,\sigma)\succsim\left(c^{\sigma''}(A),o\right)$ for some $\sigma',\sigma''\in\mathcal{E}$ with $\sigma\trianglerighteq\sigma',\sigma''$.
\end{axiom}

\begin{axiom}{\textbf{Menu-independent leakage$^\dagger$}.}
    For $A,B\in\mathcal{A}^*$ and $\sigma\in\mathcal{E}$, if there exists $\sigma^\ast\in\mathcal{E}$ such that $\sigma\trianglerighteq\sigma^\ast$ and $(A,\sigma)\sim(c^{\sigma^\ast}(A),o)$, then $(c^{\sigma^\ast}(B),o)\succsim(B,\sigma)$.
\end{axiom}

These two axioms are straightforward counterparts of \textit{betweenness} and \textit{menu-independent leakage}. The original axioms are stated in terms of information structures that are more informative than $\sigma$, whereas the modified axioms replace them with information structures that are less informative than $\sigma$.

The following corollary is obtained by slightly modifying the proof of Theorem \ref{theorem_RBP}.

\begin{corollary}
\label{corollary_garbling}
    If $\succsim$ satisfies \textit{basic rationality}, \textit{SEU with no information}, \textit{lottery invariance}, \textit{sender pessimism}, \textit{betweenness$^\dagger$}, and \textit{menu-independent leakage$^\dagger$}, then it is represented by a garbling-robust Bayesian persuasion representation.
\end{corollary}

The proof of Corollary \ref{corollary_garbling} is contained in \hyperref[Appendix]{Appendix}. 

\section*{Appendix}

\subsection*{Explicit formulation of \textit{SEU with no information}}

\textit{SEU with no information} consists of the following postulates. For all $f,g,h\in\mathcal{F}$, $\succsim$ on $\mathcal{F}\times\{o\}$ satisfies (i) \textbf{Weak order}: $\succsim$ on $\mathcal{F}\times\{o\}$ is complete and transitive, (ii) \textbf{Non-triviality}: there exist $x,y\in\Delta(X)$ such that $(x,o)\succ (y,o)$, (iii) \textbf{Independence}: for all $\alpha\in(0,1)$,  $(f,o)\succsim (g,o)$ if and only if $(\alpha f+(1-\alpha)h,o)\succsim (\alpha g+(1-\alpha)h,o)$. (iv) \textbf{Monotonicity}: If $(f(\omega),o)\succsim (g(\omega),o)$ for all $\omega\in\Omega$, then $(f,o)\succsim (g,o)$. In addition, if $(f(\omega),o)\succ (g(\omega),o)$ for some $\omega\in\Omega$, $(f,o)\succ (g,o)$. (v) \textbf{Mixture continuity}: the sets $\{\alpha\in[0,1]~|~(f,o)\succsim (\alpha g+(1-\alpha)h,o)\}$ and $\{\alpha\in[0,1]~|~(\alpha f+(1-\alpha)g, o) \succsim (h,o)\}$ are closed.

\subsection*{Proof of Theorem \ref{theorem_RBP}}\label{Appendix}

The proof consists of several lemmas.

\begin{lemma}
\label{Lemma_vNM}
    There exists a non-constant and mixture linear function $v:\Delta(X)\rightarrow\mathbb{R}$ such that for all $\sigma\in\mathcal{E}$ and $l,l'\in\Delta(X)$, 
    \begin{equation}
        (\{l\},\sigma)\succsim(\{l'\},\sigma)\iff v(l)\geq v(l').
    \end{equation}
\end{lemma}

\begin{proof}
    By \textit{SEU with no information}, there exists a non-constant and mixture linear function $v:\Delta(X)\rightarrow\mathbb{R}$ such that for all $l,l'\in\Delta(X)$, $(\{l\},o)\succsim(\{l'\},o)$ if and only if $v(l)\geq v(l')$. \textit{Lottery invariance} implies that for every $\sigma\in\mathcal{E}$, $(\{l\},\sigma)\succsim(\{l'\},\sigma)$ if and only if $(\{l\},o)\succsim(\{l'\},o)$ for all $l,l'\in\Delta(X)$. Indeed, by \textit{lottery invariance}, $(\{l\},o)\sim(\{l\},\sigma)\succsim(\{l'\},\sigma)\sim(\{l'\},o)$ holds, which implies $(\{l\},o)\succsim(\{l'\},o)$. By the same argument, $(\{l\},o)\succsim(\{l'\},o)$ implies $(\{l\},\sigma)\succsim(\{l'\},\sigma)$. Thus, for all $\sigma\in\mathcal{E}$ and $l,l'\in\Delta(X)$, $(\{l\},\sigma)\succsim(\{l'\},\sigma)\iff v(l)\geq v(l')$.
\end{proof}

\begin{lemma}{(\cite{anscombe1963definition})}
\label{Lemma_SEU}
    There exists $\nu\in\Delta(\Omega)$ such that for all acts $f,g\in\mathcal{F}$, 
    \begin{equation}
        (\{f\},o)\succsim(\{g\},o)\iff \sum_{\omega\in\Omega} v(f_\omega)\nu(\omega)\geq \sum_{\omega\in\Omega} v(g_\omega)\nu(\omega).
    \end{equation}
\end{lemma}

\begin{proof}
    This result directly comes from the result of \cite{anscombe1963definition}, thus we omit the proof.
\end{proof}

\begin{lemma}
\label{Lemma_min}
    For all $\sigma\in\mathcal{E}$ and menus $A,B\in\mathcal{A}$,
    \begin{equation}
        (c^\sigma(A),o)\succsim(c^\sigma(B),o) \iff \min_{f \in c^{\sigma} (A)} \left[ \sum_{\omega \in \Omega} v (f_\omega) \nu_\omega \right] \geq \min_{g \in c^{\sigma} (B)} \left[ \sum_{\omega \in \Omega} v (g_\omega) \nu_\omega \right]. 
    \end{equation}
\end{lemma}

\begin{proof}
    Take $\sigma\in\mathcal{E}$ and $A\in\mathcal{A}$ arbitrarily. Recall that $c^\sigma(A)$ is defined as 
    \begin{equation}
        c^{\sigma}(A)=\left\{\sum_{s\in\sigma}sf^s~|~ f^s\in c^s(A)\right\}.
    \end{equation}
    Since $A$ is finite and $\sigma$ consists of finitely many signals (columns), $c^\sigma(A)$ is also finite. By Lemma \ref{Lemma_SEU}, there exists the worst act $f^*\in c^\sigma(A)$.\footnote{Note that such $f^*$ is not necessarily unique.} That is, 
    \begin{equation}
        \sum_{\omega\in\Omega}v(f_\omega)\nu(\omega)\geq \sum_{\omega\in\Omega}v(f^*_\omega)\nu(\omega)
    \end{equation}
    for all $f\in c^\sigma(A)$. We show that $(\{f^*\},o)\sim(c^\sigma(A),o)$ holds. Suppose not. Then, either $(\{f^*\},o)\succ (c^\sigma(A),o)$ or $(c^\sigma(A),o)\succ (\{f^*\},o)$ holds. If the former holds, then $(\{f\},o)\succ (c^\sigma(A),o)$ for all $f\in c^\sigma(A)$, contradicting \textit{sender pessimism}. If the latter holds, then there exists $f\in c^\sigma(A)$ such that $(c^\sigma(A),o)\succ (\{f\},o)$ , contradicting \textit{sender pessimism} again. Thus, $(\{f^*\},o)\sim(c^\sigma(A),o)$ holds. Then, for all $A,B\in\mathcal{A}$, 
    \begin{align}
        &~~~~~~~~~(c^\sigma(A),o)\succsim(c^\sigma(B),o) \\ 
        &\iff (\{f^*\},o)\succsim(\{g^*\},o) \text{, where $(\{f^*\},o)\sim(c^\sigma(A),o)$ and  $(\{g^*\},o)\sim(c^\sigma(B),o)$} \\
        &\iff \min_{f \in c^{\sigma} (A)} \left[ \sum_{\omega \in \Omega} v (f_\omega) \nu_\omega \right] \geq \min_{g \in c^{\sigma} (B)} \left[ \sum_{\omega \in \Omega} v (g_\omega) \nu_\omega \right]. 
    \end{align}
    The first equivalence comes from \textit{basic rationality} and the second one comes from the construction of $f^*$ and $g^*$.
\end{proof}

\begin{lemma}
    For all $\sigma\in\mathcal{E}$ and $A\in\mathcal{A}$, there exists $l_A$ such that $(A,\sigma)\sim(\{l_A\},o)$.
\end{lemma}

\begin{proof}
    Take $\sigma\in\mathcal{E}$ and $A\in\mathcal{A}$ arbitrarily. By \textit{betweenness}, there exist $\sigma',\sigma''\in\mathcal{E}$ such that $\sigma',\sigma''\trianglerighteq\sigma$ and 
    \begin{equation}
        \left(c^{\sigma'}(A),o\right)\succsim(A,\sigma)\succsim\left(c^{\sigma''}(A),o\right).
    \end{equation}
    By Lemma \ref{Lemma_min}, there exist lotteries $l',l''\in\Delta(X)$ such that 
    \begin{equation}
        \min_{f \in c^{\sigma'} (A)} \left[ \sum_{\omega \in \Omega} v (f_\omega) \nu_\omega \right]=v(l') \text{ and } \min_{f \in c^{\sigma''} (A)} \left[ \sum_{\omega \in \Omega} v (f_\omega) \nu_\omega \right]=v(l'').
    \end{equation}
    Such lotteries exist since $v$ is a mixture-linear function. By \textit{basic rationality}, 
    \begin{equation}
        (\{l'\},o)\succsim (A,\sigma)\succsim (\{l''\},o).
    \end{equation}
    Moreover, by \textit{lottery invariance} and \textit{basic rationality}, 
    \begin{equation}
        (\{l'\},\sigma)\succsim (A,\sigma)\succsim (\{l''\},\sigma).
    \end{equation}
    Then, \textit{basic rationality} implies that the sets 
    \begin{equation}
       I=\{\alpha\in[0,1]~|~(\{\alpha l'+(1-\alpha)l''\},\sigma)\succsim(A,\sigma)\}
    \end{equation}
    and 
    \begin{equation}
       J=\{\alpha\in[0,1]~|~(A,\sigma)\succsim(\{\alpha l'+(1-\alpha)l''\},\sigma)\}
    \end{equation}
    are closed.
    By \textit{basic rationality}, $I\cup J=[0,1]$. Since $[0,1]$ is connected, there exists $\alpha\in I\cap J$ such that $(\{\alpha l'+(1-\alpha)l''\},\sigma)\sim(A,\sigma)$. By \textit{lottery invariance}, $(\{\alpha l'+(1-\alpha)l''\},o)\sim(A,\sigma)$ holds and the proof ends.
\end{proof}

\begin{lemma}
\label{lemma_final}
    For every $\sigma\in\mathcal{E}$, there exists $\mathcal{L}(\sigma)\subset \{\sigma'\in\mathcal{E}~|~\sigma'\trianglerighteq\sigma\}$ such that 
    \begin{equation}
        v(l_A)=\inf_{\pi\in \mathcal{L}(\sigma)} \sum_{s \in \pi} \left[ \min_{f^s \in c^s (A)} \sum_{\omega \in \Omega} v (f_\omega) s_\omega \nu_\omega \right].
    \end{equation}
\end{lemma}

\begin{proof}
    \textit{Betweenness} implies that for all $A\in\mathcal{A}$, there exists $\sigma^*$ such that $\sigma^*\trianglerighteq\sigma$ and 
    \begin{equation}
        \min_{f \in c^{\sigma^*} (A)} \left[ \sum_{\omega \in \Omega} v (f_\omega) \nu_\omega \right] \geq v(l_A).
    \end{equation}
    Moreover, \textit{betweenness} implies that
    \begin{equation}
        v(l_A)\geq \min_{\sigma'\trianglerighteq\sigma} \sum_{s \in \sigma'} \left[ \min_{f^s \in c^s (A)} \sum_{\omega \in \Omega} v (f_\omega) s_\omega \nu_\omega \right]
    \end{equation}
    holds.
    Then, there exists $\sigma_*\in\mathcal{E}$ such that $\sigma_*\trianglerighteq\sigma$ and 
    \begin{equation}
        v(l_A)\geq \min_{f \in c^{{\sigma_*}} (A)} \left[ \sum_{\omega \in \Omega} v (f_\omega) \nu_\omega \right].
    \end{equation}
    Notice that for all $\alpha\in[0,1]$,
    \begin{equation}
        c^{\alpha\sigma^*\cup(1-\alpha)\sigma_*}(A)=\left\{\alpha\sum_{s\in\sigma^*}sf^s+(1-\alpha)\sum_{s'\in\sigma_*}s'f^{s'}~|~ f^s\in c^s(A),~ f^{s'}\in c^{s'}(A)\right\}.
    \end{equation}
    Moreover, since $v$ is mixture-linear, 
    \begin{equation}
        \min_{f \in c^{\alpha\sigma^*\cup(1-\alpha)\sigma_*}(A)} \left[ \sum_{\omega \in \Omega} v (f_\omega) \nu_\omega \right] = \alpha \min_{f \in c^{{\sigma^*}} (A)} \left[ \sum_{\omega \in \Omega} v (f_\omega) \nu_\omega \right] + (1-\alpha) \min_{f \in c^{{\sigma_*}} (A)} \left[ \sum_{\omega \in \Omega} v (f_\omega) \nu_\omega \right].
    \end{equation}
    Since $\alpha\sigma^*\cup(1-\alpha)\sigma_*$ is Blackwell more informative than $\sigma$ for all $\alpha\in[0,1]$, there exists $\sigma^A\in\mathcal{E}$ such that $\sigma^A\trianglerighteq\sigma$ and
    \begin{equation}
        v(l_A)=\min_{f \in c^{\sigma^A} (A)} \left[ \sum_{\omega \in \Omega} v (f_\omega) \nu_\omega \right].
    \end{equation}
    Define the function 
    $\phi:\mathcal{A}\times\mathcal{E}\rightarrow\mathcal{E}$ such that for all $A\in\mathcal{A}$ and $\sigma\in\mathcal{E}$, $\phi(A,\sigma)=\sigma^A$. Notice that if a menu $A$ is either a set of lotteries or a singleton, without loss of generality, we can set $\phi(A,\sigma)=\sigma^A=\sigma$. Also we can set $\phi(A,\sigma)=\sigma^A=\sigma$ if a menu $A$ satisfies 
    \begin{equation}
        \min_{f \in c^{\phi(A,\sigma)} (A)} \left[ \sum_{\omega \in \Omega} v (f_\omega) \nu_\omega \right] = \min_{f \in c^{\sigma} (A)} \left[ \sum_{\omega \in \Omega} v (f_\omega) \nu_\omega \right].
    \end{equation}
    Then, define 
    \begin{equation}
        \mathcal{L}(\sigma)=\left\{\phi(A,\sigma)~|~ \min_{f \in c^{\phi(A,\sigma)} (A)} \left[ \sum_{\omega \in \Omega} v (f_\omega) \nu_\omega \right] \neq \min_{f \in c^{\sigma} (A)} \left[ \sum_{\omega \in \Omega} v (f_\omega) \nu_\omega \right],A\in\mathcal{A}^* \right\}~\bigcup~\{\sigma\}.
    \end{equation}
    We show that 
    \begin{equation}
        \min_{f \in c^{\phi(A,\sigma)} (A)} \left[ \sum_{\omega \in \Omega} v (f_\omega) \nu_\omega \right] =\inf_{\sigma'\in\mathcal{L}(\sigma)} \min_{f \in c^{\sigma'} (A)} \left[ \sum_{\omega \in \Omega} v (f_\omega) \nu_\omega \right].
    \end{equation}
    By the construction of $\mathcal{L}(\sigma)$, 
    \begin{equation}
        \min_{f \in c^{\phi(A,\sigma)} (A)} \left[ \sum_{\omega \in \Omega} v (f_\omega) \nu_\omega \right] \geq\inf_{\sigma'\in\mathcal{L}(\sigma)} \min_{f \in c^{\sigma'} (A)} \left[ \sum_{\omega \in \Omega} v (f_\omega) \nu_\omega \right]
    \end{equation}
    holds.
    Suppose to the contrary that the strict inequality
    \begin{equation}
        \min_{f \in c^{\phi(A,\sigma)} (A)} \left[ \sum_{\omega \in \Omega} v (f_\omega) \nu_\omega \right] > \inf_{\sigma'\in\mathcal{L}(\sigma)} \min_{f \in c^{\sigma'} (A)} \left[ \sum_{\omega \in \Omega} v (f_\omega) \nu_\omega \right]
    \end{equation}
    holds.
    Then, there exists $\sigma^\ast\in\mathcal{L}(\sigma)$ such that
    \begin{equation}
        \min_{f \in c^{\phi(A,\sigma)} (A)} \left[ \sum_{\omega \in \Omega} v (f_\omega) \nu_\omega \right] >\min_{f \in c^{\sigma^\ast} (A)} \left[ \sum_{\omega \in \Omega} v (f_\omega) \nu_\omega \right].
    \end{equation}
    By the construction of $\mathcal{L}(\sigma)$, there exists a menu $B\in\mathcal{A}^*$ such that $(B,\sigma)\sim(c^{\sigma^\ast}(B),o)$. The above inequality implies $(A,\sigma)\succ(c^{\sigma^\ast}(A),o)$, a contradiction to \textit{menu-independent leakage}.
\end{proof}

\subsection*{Proof of Proposition \ref{proposition_RBP_nec}}

It is clear that any robust Bayesian persuasion representation (not necessarily inclusive $\mathcal{L}$) satisfies \textit{basic rationality}, \textit{SEU with no information}, \textit{lottery invariance}, \textit{sender pessimism}, and \textit{betweenness}.
Thus, we only show that if $(\succsim,c)$ is represented by a robust Bayesian persuasion representation and $\mathcal{L}$ is inclusive, then it satisfies \textit{menu-independent leakage}. 
Take any $A,B\in\mathcal{A}^*$ and $\sigma\in\mathcal{E}$ and suppose that there exists $\sigma^\ast\in\mathcal{E}$ such that $\sigma^\ast\trianglerighteq\sigma$ and $(A,\sigma)\sim(c^{\sigma^\ast}(A),o)$. Since $\mathcal{L}$ is inclusive, $\sigma^*\in\mathcal{L}(\sigma)$.
Thus, 
\begin{equation}
    \min_{f \in c^{\sigma^\ast}(B)} \left[ \sum_{\omega \in \Omega} v (f_\omega) \nu_\omega \right]\geq \inf_{\pi\in \mathcal{L}(\sigma)} \sum_{s \in \pi} \left[ \min_{f^s \in c^s (B)} \sum_{\omega \in \Omega} v (f_\omega) s_\omega \nu_\omega \right]= V(B,\sigma).
\end{equation}
This is equivalent to $(c^{\sigma^\ast}(B),o)\succsim(B,\sigma)$.

\subsection*{Proof of Proposition \ref{proposition_uniqueness}}

It is routine to show that there exist $\alpha>0$ and $\beta\in\mathbb{R}$ such that $v'=\alpha v+\beta$, $\nu'=\nu$. We show that two maximal functions satisfy  $\bar{\mathcal{L}}=\bar{\mathcal{L}}'$. Suppose to the contrary that there exists $\sigma\in\mathcal{E}$ such that $\bar{\mathcal{L}}(\sigma)\neq \bar{\mathcal{L}}'(\sigma)$. Then, without loss of generality, there exists $\sigma^*\in\bar{\mathcal{L}}(\sigma)$ such that $\sigma^*\notin \bar{\mathcal{L}}'(\sigma)$. By the definition of $\bar{\mathcal{L}}$, $V\left(c^{\sigma^*}(A),o\right)\geq V(A,\sigma)$ for all $A\in\mathcal{A}^*$. On the other hand, since $\sigma^*\notin \bar{\mathcal{L}}'(\sigma)$, there exists $B\in\mathcal{A}^*$ such that $V(B,\sigma)>V\left(c^{\sigma^*}(B),o\right)$. This is a contradiction to $V\left(c^{\sigma^*}(B),o\right)\geq V(B,\sigma)$. Thus, $\bar{\mathcal{L}}=\bar{\mathcal{L}}'$.

\subsection*{Proof of Proposition \ref{proposition_comparative}}

We first show that the first condition implies the second condition. Suppose to the contrary that for some $\sigma\in\mathcal{E}$, $\bar{\mathcal{L}}_1(\sigma)\not\subseteq \bar{\mathcal{L}}_2(\sigma)$. Then, there exists $\sigma'\in \bar{\mathcal{L}}_1(\sigma)$ such that $\sigma'\notin \bar{\mathcal{L}}_2(\sigma)$. Then, by the definition of $\bar{\mathcal{L}}_2$, there exists $A\in\mathcal{A}^*$ such that $(A,\sigma) \succ_2 \left(c^{\sigma'}(A),o\right)$ holds. On the other hand, since $\sigma'\in \bar{\mathcal{L}}_1(\sigma)$, for all $B\in\mathcal{A}^*$, $\left(c^{\sigma'}(B),o\right)\succsim_1 (B,\sigma)$. Since sender 1's vNM function and prior belief coincide with those of sender 2, there exists a lottery $l\in\Delta
(X)$ such that $\left(c^{\sigma'}(A),o\right)\sim_1 (\{l\},o)$ and $\left(c^{\sigma'}(A),o\right)\sim_2 (\{l\},o)$. This implies that $(\{l\},o)\succsim_1 (A,\sigma)$ but $(A,\sigma) \succ_2 (\{l\},o)$. This contradicts the first condition. 

The proof that the second condition implies the first condition is straightforward since for all $A\in\mathcal{A}$, $l\in\Delta(X)$, and $\sigma\in\mathcal{E}$,
\begin{equation}
    V_1(\{l\},\sigma) \geq \inf_{\pi\in \bar{\mathcal{L}}_1(\sigma)} \sum_{s\in\pi,\omega\in\Omega} \underline{V}_1 (\nu ^s) s_\omega \nu_\omega \geq \inf_{\pi\in \bar{\mathcal{L}}_2(\sigma)} \sum_{s\in\pi,\omega\in\Omega} \underline{V}_2 (\nu ^s) s_\omega \nu_\omega.
\end{equation}

\subsection*{Proof of Corollary \ref{corollary_information_avoidance}}

Take any $A\in\mathcal{A}$ and $\sigma\in\mathcal{E}$. Let $l_A\in\Delta(X)$ be a lottery such that $(\{l_A\},o)\sim_i (A,o)$ for $i\in\{1,2\}$. Notice that such a lottery exists since sender 1's vNM function and prior belief are coincide with those of sender 2. Then, applying Proposition \ref{proposition_comparative}, we obtain 
\begin{equation}
    (\{l_A\},o)\succsim_1 (A,\sigma) \implies (\{l_A\},o)\succsim_2 (A,\sigma).
\end{equation}
Thus, 
\begin{equation}
    (A,o)\succsim_1 (A,\sigma) \implies (A,o)\succsim_2 (A,\sigma).
\end{equation}

\subsection*{Proof of Corollary \ref{corollary_BP}}

The proof of Corollary \ref{corollary_BP} follows from a slight modification of Lemma \ref{Lemma_min} in Theorem \ref{theorem_RBP}. By an argument analogous to that in Lemma \ref{Lemma_min}, \textit{sender optimism} implies that for all $A,B\in\mathcal{A}$ and $\sigma\in\mathcal{E}$, 
\begin{equation}
    (c^\sigma(A),o)\succsim(c^\sigma(B),o) \iff \max_{f \in c^{\sigma} (A)} \left[ \sum_{\omega \in \Omega} v (f_\omega) \nu_\omega \right] \geq \max_{g \in c^{\sigma} (B)} \left[ \sum_{\omega \in \Omega} v (g_\omega) \nu_\omega \right]. 
\end{equation}
Moreover, by \textit{reduction}, this representation directly extends to the sender's preference over pairs consisting of a menu and the information structure $\sigma$.
By algebra,
\begin{equation}
    \max_{f \in c^{\sigma} (A)} \left[ \sum_{\omega \in \Omega} v (f_\omega) \nu_\omega \right]=\sum_{s\in\sigma,\omega\in\Omega} \left[ \max_{f\in c^s(A)} \sum_{\omega\in\Omega} v(f_\omega) \nu^s_\omega \right]s_\omega \nu_\omega
\end{equation}
holds. Finally, by \textit{basic rationality}, for any $A, B$ and $\sigma, \sigma'$,
\begin{align}
    &~~~~~~~~~(A,\sigma)\succsim(B,\sigma') \\
    &\iff \sum_{s\in\sigma,\omega\in\Omega} \left[ \max_{f\in c^s(A)} \sum_{\omega\in\Omega} v(f_\omega) \nu^s_\omega \right]s_\omega \nu_\omega \geq \sum_{s'\in\sigma',\omega\in\Omega} \left[ \max_{f\in c^{s'}(A)} \sum_{\omega\in\Omega} v(f_\omega) \nu^{s'}_\omega \right]{s'}_\omega \nu_\omega
\end{align}
holds.

\subsection*{Proof of Proposition \ref{proposition_fact-checking}}

The argument follows the proof of Theorem \ref{theorem_RBP}.
The key observation is that, for any $p,q,\alpha\in[0,1]$,
\[
\alpha\pi_p(\sigma)\cup(1-\alpha)\pi_q(\sigma)
\]
is Blackwell equivalent to
\[
\pi_{\alpha p+(1-\alpha)q}(\sigma).
\]
Indeed, merging proportional columns yields the latter experiment.
Thus, the interpolation step in that proof can be carried out within the family $\{\pi_p(\sigma):p\in[0,1]\}$.
In particular, for each $(A,\sigma)$, we can choose
$\phi(A,\sigma)=\pi_p(\sigma)$ for some $p\in[0,1]$ such that
\[
(A,\sigma)\sim(c^{\phi(A,\sigma)}(A),o).
\]
The remaining construction, using \textit{menu-independent leakage}, proceeds as in Theorem \ref{theorem_RBP}. Since $\pi_0(\sigma)=\sigma$, adding $\sigma$ to the constructed set also preserves this form. Collecting the corresponding weights into $P(\sigma)$ gives the representation.

\subsection*{Proof of Corollary \ref{corollary_garbling}}

Up to Lemma \ref{Lemma_min}, the same argument with the proof of Theorem \ref{theorem_RBP} holds. By \textit{betweenness$^\dagger$}, for all $\sigma\in\mathcal{E}$ and $A\in\mathcal{A}$, there exists a lottery $l_A$ such that $(A,\sigma)\sim (\{l_A\},o)$. Moreover, by the same argument with Lemma \ref{lemma_final}, there exist $\sigma^*,\sigma_*\in\mathcal{E}$ such that $\sigma\trianglerighteq\sigma^*,\sigma_*$ and 
\begin{equation}
    v(l_A) = \min_{f \in c^{\alpha\sigma^*\cup(1-\alpha)\sigma_*}(A)} \left[ \sum_{\omega \in \Omega} v (f_\omega) \nu_\omega \right]
\end{equation}
for some $\alpha\in[0,1]$.
Let $\phi(A,\sigma)=\alpha\sigma^*\cup(1-\alpha)\sigma_*$ and define 
\begin{equation}
    \mathcal{L^\dagger}(\sigma)=\left\{\phi(A,\sigma)~|~ \min_{f \in c^{\phi(A,\sigma)} (A)} \left[ \sum_{\omega \in \Omega} v (f_\omega) \nu_\omega \right] \neq \min_{f \in c^{\sigma} (A)} \left[ \sum_{\omega \in \Omega} v (f_\omega) \nu_\omega \right],A\in\mathcal{A}^* \right\}~\bigcup~\{\sigma\}.
\end{equation}
By the same argument with Lemma \ref{lemma_final}, \textit{menu-independent leakage$^\dagger$} implies that 
\begin{equation}
    v(l_A)=\inf_{\pi\in \mathcal{L^\dagger}(\sigma)} \sum_{s \in \pi} \left[ \min_{f^s \in c^s (A)} \sum_{\omega \in \Omega} v (f_\omega) s_\omega \nu_\omega \right].
\end{equation}
Thus, the proof ends.

\bibliographystyle{econ}
\bibliography{Literature}

\end{document}